\documentclass[a4paper,USenglish]{lipics-v2016}
 
\usepackage{microtype}
\usepackage{amsmath}

\theoremstyle{plain}
\newtheorem{proposition}[theorem]{Proposition}

\usepackage{todonotes}
\usepackage{fancyhdr}
\usepackage{hyperref}
\usepackage[ruled,linesnumbered]{algorithm2e}
\usepackage{babel}
\usepackage{float}
\usepackage{mathtools}
 \usepackage{subcaption}
\usepackage{wrapfig}

\bibstyle{abbrv}

\title{Faster Betweenness Centrality Updates in Evolving Networks\footnote{This work was partially supported by DFG grant ME-3619/3-1 (FINCA) and Br 2158/11-1 within the SPP 1736 \emph{Algorithms for Big Data}. A. S. acknowledges support by the RISE program of
DAAD.}}

\author[1]{Elisabetta Bergamini}
\author[1]{Henning Meyerhenke}
\author[2]{Mark Ortmann}
\author[3]{Arie Slobbe}
\affil[1]{Karlsruhe Institute of Technology (KIT), Germany\\
  \texttt{\{elisabetta.bergamini, meyerhenke\}\,@\,kit.edu}}
 \affil[2]{University of Konstanz, Germany\\
  \texttt{mark.ortmann\,@\,uni-konstanz.de}}
  \affil[3]{Australian National University, Australia\\
  \texttt{arieslobbe1\,@\,gmail.com}}
\authorrunning{E. Bergamini, H. Meyerhenke, M. Ortmann and A. Slobbe} 

\Copyright{Elisabetta Bergamini, Henning Meyerhenke, Mark Ortmann and Arie Slobbe}


\newcommand{\etal}{et al.\xspace}

\newcount\shortyear\newcount\shorthour\newcount\shortminute
\shorthour=\time\divide\shorthour by 60\shortyear=\shorthour
\multiply\shortyear by 60\shortminute=\time\advance\shortminute by
-\shortyear
\shortyear=\year\advance\shortyear by -1900

\def\zeit{\number\shorthour:\ifnum\shortminute<10 0\number\shortminute
\else\number\shortminute\fi}

\def\rightmark{{\small \today, \zeit{}}}

\date{\today}

\begin{document}
\maketitle

\begin{abstract}
Finding central nodes is a fundamental problem in network analysis. Betweenness centrality is a well-known measure which quantifies the importance of a node based on the fraction of shortest paths going though it.
Due to the dynamic nature of many today's networks, algorithms that quickly update centrality scores have become a necessity. 
For betweenness, several dynamic algorithms have been proposed over the years, targeting different update types (incremental- and decremental-only, fully-dynamic).
In this paper we introduce a new dynamic algorithm for updating betweenness centrality after an edge insertion or an edge weight decrease. 
Our method is a combination of two independent contributions: a faster algorithm for updating pairwise distances as well as number of shortest paths, and a faster algorithm for updating dependencies. Whereas the worst-case running time of our algorithm is the same as recomputation, our techniques considerably reduce the number of operations performed by existing dynamic betweenness algorithms.
Our experimental evaluation on a variety of real-world networks reveals that our approach is significantly faster than the current state-of-the-art dynamic algorithms, approximately by one order of magnitude on average.
\end{abstract}
\section{Introduction}
Over the last years, increasing attention has been devoted to the analysis of complex networks. 
A common sub-problem for many graph based applications is to identify the most central nodes in a network. 
Examples include facility location~\cite{Koschutzki:2005fk}, marketing strategies~\cite{Kiss2008233} and identification of key infrastructure nodes 
as well as disease propagation control and crime prevention~\cite{DBLP:journals/socnet/BellAC99}.  As the meaning of ``central'' heavily depends on the context, various centrality measures have been proposed (see~\cite{DBLP:journals/im/BoldiV14} for an overview). Betweenness centrality is a well-known measure which ranks nodes according to their participation in the shortest paths of the network. Formally, the betweenness of a node $v$ is defined as $c_B(v) = \sum_{s \neq v \neq t} \frac{\sigma_{st}(v)}{\sigma_{st}}$, where $\sigma_{st}$ is the number of shortest paths between two nodes $s$ and $t$ and $\sigma_{st}(v)$ is the number of these paths that go through node $v$. The fastest algorithm for computing betweenness centrality is due to Brandes~\cite{Brandes01betweennessCentrality}, which we refer to as \textsf{BA}, from Brandes's algorithm. This algorithm is composed of two parts: an augmented APSP (all-pairs shortest paths) step, where pairwise distances and shortest paths are computed, and a dependency accumulation step, where the actual betweenness scores are computed. 
The augmented APSP is computed by running a SSSP (single-source shortest paths) computation from each node $s$ and the dependency accumulation is performed by traversing only once the edges that lie in shortest paths between $s$ and the other nodes. 
Therefore, \textsf{BA} requires $\Theta(|V| |E|)$ time on unweighted and $\Theta(|V| |E| + |V|^2 \log |V|)$ time on weighted graphs (i.e. the time of running $n$ SSSPs).

Networks such as the Web graph and social networks continuously undergo changes. Since an update in the graph might affect only a small fraction of nodes, recomputing betweenness with \textsf{BA} after each update would be very inefficient. For this reason, several dynamic algorithms have been proposed over the last years~\cite{DBLP:conf/socialcom/GreenMB12,kourtellis2014scalable,DBLP:conf/asunam/KasWCC13}. 
As \textsf{BA}, these approaches usually solve two sub-tasks: the update of the augmented APSP data structures and the update of the betweenness scores.
Although none of these algorithms is in general asymptotically faster than recomputation with \textsf{BA}, good speedups over \textsf{BA} have been reported for some of them, in particular for~\cite{DBLP:conf/asunam/KasWCC13} and~\cite{kourtellis2014scalable}. Nonetheless, an exhaustive comparison of these methods is missing in the literature.

In our paper, we only consider \textit{incremental} updates, i.e. edge insertions or edge weight decreases (node insertions can be handled treating the new node as an isolated node and adding its neighboring edges one by one). Although it might seem reductive to only consider these kinds of updates, it is important to note that several real-world dynamic networks evolve only this way and do not shrink. For example, in a co-authorship network, a new author (node) or a new edge (coauthored publication) might be added to the network, but existing nodes or edges will not disappear. Another possible application is the centrality maximization problem, which consists in finding a set of edges that, if added to the graph, would maximize the centrality of a certain node. The problem can be approximated with a heuristic~\cite{DBLP:conf/wea/CrescenziDSV15}, which requires to add several edges to the graph and to recompute distances after each edge insertion.


\paragraph*{Our contribution}
We present a new algorithm for updating betweenness centrality after an edge insertion or an edge weight decrease.
Our method is a combination of two contributions: a new dynamic algorithm for the augmented APSP, and a new approach for updating the betweenness scores. Based on properties of the newly-created shortest paths, our dynamic APSP algorithm efficiently identifies the node pairs affected by the edge update (i.e. those for which the distance and/or number of shortest paths change as a consequence of the update). The betweenness update method works by accumulating values in a fashion similar to that of \textsf{BA}. However, differently from \textsf{BA}, our method only processes nodes that lie in shortest paths between affected pairs. 

We compare our new approach with two of the dynamic algorithms for which the best speedups over recomputation have been reported in the literature, i.e. \textsf{KWCC}~\cite{DBLP:conf/asunam/KasWCC13} and \textsf{KDB}~\cite{kourtellis2014scalable}. Compared to them, our algorithm for the augmented APSP update is asymptotically faster on dense graphs: $O(|V|^2)$ in the worst case versus $O(|V||E|)$. This is due to the fact that we iterate over the edges between affected nodes only once, whereas \textsf{KDB} and \textsf{KWCC} do it several times. Moreover, our dependency update works also for weighted graphs (whereas \textsf{KDB} does not) and it is asymptotically faster than the dependency update of \textsf{KWCC} for sparse graphs ($O(|V||E| + |V| \log |V|)$ in the worst case versus $O(|V|^3)$).

Our experimental evaluation on a variety of real-world networks reveals that our approach is significantly faster than both \textsf{KDB} and \textsf{KWCC}, on average by a factor 14.7 and 7.4, respectively.

\section{Preliminaries}
\subsection{Notation}
Let $G = (V, E, \omega)$ be a graph with node set $V = V(G)$, edge set $E = E(G)$  and edge weights $\omega: E \rightarrow \mathbb{R}_{> 0}$. In the following we will use $n := |V|$ to denote the number of nodes and $m := |E|$ for the number of edges.
Let $d(s, t)$ be the shortest-path distance between any two nodes $s, t \in V$. On a shortest path from $s$ to $t$ in $G$, we say $w$ is a \textit{predecessor} of $t$, or $t$ is a \textit{successor} of $w$, if $(w,t) \in E$ and $d(s,w) + \omega(w,t) = d(s,t)$. We denote the set of predecessors of $t$ as $P_s(t)$.
For a given source node $s \in V$, we call the graph composed of the nodes reachable from $s$ and the edges that lie in at least one shortest path from $s$ to any other node the \emph{SSSP DAG of $s$}.
We use $\sigma_{st}$ to denote the number of shortest paths between $s$ and $t$ and we use $\sigma_{st}(v)$ for the number of shortest paths between $s$ and $t$ that go through $v$. Then, the betweenness centrality $c_B(v)$ of a node $v$ is defined as: $c_B(v) = \sum_{s \neq v \neq t} \frac{\sigma_{st}(v)}{\sigma_{st}}$.

Our goal is to keep track of the betweenness scores of all nodes after an update $(u, v, \omega'(u, v))$ in the graph, which could either be an edge insertion or an edge weight decrease. We use $G' = (V, E', \omega')$ to denote the new graph after the edge update and $d'$, $\sigma'$ and $P'$ to denote the new distances, numbers of shortest paths and sets of predecessors, respectively. Also, we define the set of \textit{affected sources} $S(t)$ of a node $t\in V$ as $\{ s \in V : d(s, t) > d'(s, t) \vee \sigma_{st} \neq \sigma'_{st}\}$. Analogously, we define the set of affected targets of $s \in V$ as $T(s) := \{ t \in V : d(s, t) > d'(s, t) \vee \sigma_{st} \neq \sigma'_{st}\}$.
In the following we will assume $G$ to be directed. However, the algorithms can be easily extended to undirected graphs.

\subsection{Related Work}
The basic idea of dynamic betweenness algorithms is to keep track of the old betweenness scores (and additional data structures) and efficiently update the information after some modification in the graph. 
Based on the type of updates they can handle, dynamic algorithms are classified as \textit{incremental} (only edge insertions and weight decreases), \textit{decremental} (only edge deletions and weight increases) or \textit{fully-dynamic} (all kinds of edge updates).
However, one commonality of all these approaches is that they build on the techniques used by \textsf{BA}~\cite{Brandes01betweennessCentrality}, which we therefore describe in Section~\ref{sec:brandes} in more detail.

The approach proposed
by Green \etal~\cite{DBLP:conf/socialcom/GreenMB12} for unweighted graphs 
maintains all previously calculated betweenness values and
additional information, such as pairwise distances, number of shortest paths and lists of predecessors of each node in the shortest paths from each source node $s \in V$. Using this information,
the algorithm tries to limit the recomputation to the nodes whose
betweenness has been affected by the edge insertion. Kourtellis
\etal~\cite{kourtellis2014scalable} modify the
approach by Green \etal~\cite{DBLP:conf/socialcom/GreenMB12} in
order to reduce the memory requirements from $O(nm)$ to $O(n^2)$. 
Instead of being stored, the predecessors are recomputed every time the algorithm requires them. The authors show that not only using less memory allows them to scale to larger graphs, but their approach (which we refer to as \textsf{KDB}, from the authors's initials) turns out to be also faster than the one by Green \etal~\cite{DBLP:conf/socialcom/GreenMB12} in practice (most likely because of the cost of maintaining the data structure of the algorithm by Green \etal).

Kas \etal~\cite{DBLP:conf/asunam/KasWCC13} extend an existing algorithm for
the dynamic all-pairs shortest paths (APSP) problem by Ramalingam and Reps~\cite{DBLP:journals/tcs/RamalingamR96}
to also update betweenness scores. Differently from the previous two approaches, this algorithm can handle also weighted graphs. Although good speedups have been reported for this approach, no experimental evaluation compares its performance with that of the approaches by Green \etal~\cite{DBLP:conf/socialcom/GreenMB12} and  Kourtellis \etal~\cite{kourtellis2014scalable}. We refer to this algorithm as \textsf{KWCC}, from the authors's initials.

Nasre \etal~\cite{DBLP:conf/mfcs/NasrePR14} compare the distances between each node pair before and after the update and then recompute the dependencies from scratch as in \textsf{BA} (see Section~\ref{sec:brandes}). Although this algorithm is faster than recomputation on some graph classes (i.e.\ when only edge insertions are allowed and the graph is sparse and weighted), it was shown in~\cite{DBLP:conf/alenex/BergaminiMS15} that its practical performance is much worse than that of the algorithm proposed by Green \etal~\cite{DBLP:conf/socialcom/GreenMB12}. This is quite intuitive, since recomputing all dependencies requires $\Omega(n^2)$ time independently of the number of nodes that are actually affected by the insertion.

Pontecorvi and Ramachandran~\cite{DBLP:conf/isaac/PontecorviR15} extend existing fully-dynamic APSP algorithms with new data structures to update \textit{all} shortest paths and then recompute dependencies as in \textsf{BA}. To our knowledge, this algorithm has never been implemented, probably because of the quite complicated data structures it requires. Also, since it recomputes dependencies from scratch as Nasre \etal~\cite{DBLP:conf/mfcs/NasrePR14}, we expect its practical performance to be similar.

Differently from the other algorithms, the approach by Lee \etal~\cite{DBLP:journals/isci/LeeCC16} is not based on dynamic APSP algorithms. The idea is to decompose the graph into its biconnected components and then recompute the betweenness values from scratch only for the nodes in the component affected by the update. Although this allows for a smaller memory requirement ($\Theta(m)$ versus $\Omega(n^2)$ needed by the other approaches), the speedups on recomputation reported in~\cite{DBLP:journals/isci/LeeCC16} are significantly worse than those reported for example by Kourtellis \etal~\cite{kourtellis2014scalable}.

To summarize, \textsf{KDB}~\cite{kourtellis2014scalable} and \textsf{KWCC}~\cite{DBLP:conf/asunam/KasWCC13} are the most promising methods for a comparison with our new algorithm. For this reason, we will describe them in more detail in Section~\ref{sec:dyn-apsp} and Section~\ref{sec:dependency} and evaluate them in our experiments.

Since computing betweenness exactly can be too expensive for large networks, several approximation algorithms and heuristics have been introduced in the literature~\cite{DBLP:conf/esa/BorassiN16,DBLP:conf/alenex/GeisbergerSS08,DBLP:journals/datamine/RiondatoK16,DBLP:conf/kdd/RiondatoU16} and, recently, also dynamic algorithms that update an approximation of betweenness centrality have been proposed~\cite{DBLP:journals/im/BergaminiM16,DBLP:conf/alenex/BergaminiMS15,DBLP:journals/pvldb/HayashiAY15,DBLP:conf/kdd/RiondatoU16}.
However, we will not consider them in our experimental evaluation since our focus here is on exact methods.

\section{Brandes's algorithm (\textsf{BA})}
\label{sec:brandes}
Betweenness centrality can be easily computed in time $\Theta(n^3)$ by simply applying its definition. In 2001, Brandes proposed an algorithm (\textsf{BA})~\cite{Brandes01betweennessCentrality} which requires time $\Theta(nm)$ for unweighted and $\Theta(n(m + n\log n))$ for weighted graphs, i.e. the time of computing $n$ single-source shortest paths (SSSPs).
The algorithm is composed of two parts: the \textit{augmented APSP} computation phase based on $n$ SSSPs and the \textit{dependency accumulation} phase. As dynamic algorithms based on \textsf{BA} build on these two steps as well, we explain them now in more detail.
\paragraph*{Augmented APSP} In this first part, \textsf{BA} needs to perform an \textit{augmented} APSP, meaning that instead of simply computing distances between all node pairs $(s, t)$, it also finds the number of shortest paths $\sigma_{st}$ and the set of predecessors $P_s(t)$. This can be done while computing an SSSP from each node $s$ (i.e. BFS for unweighted and Dijkstra for weighted graphs). When a node $w$ is extracted from the SSSP (priority) queue, \textsf{BA} computes $P_s(w)$ as $\{v : (v, w) \in E \, \wedge \, d(s,w) = d(s, v) + \omega(v,w)\}$ and $\sigma_{sw}$ as $\sum_{v \in P_s(w)} \sigma_{sv}$.

\paragraph*{Dependency accumulation}
Brandes defines the one-side dependency of a node $s$ on a node $v$ as $\delta_{s \bullet}(v) := \sum_{t \neq v} \sigma_{st}(v)/ \sigma_{st}$.
It can be proven~\cite{Brandes01betweennessCentrality} that
\begin{equation}
\label{eq:dep}
\delta_{s \bullet}(v) = \sum_{w : v \in P_s(w)} \frac{\sigma_{sv}}{\sigma_{sw}} (1 + \delta_{s \bullet}(w)), \ \ \ \ \forall s, v \in V
\end{equation}
Intuitively, the term $ \delta_{s \bullet}(w)$ in Eq.~(\ref{eq:dep}) represents the contribution of the sub-DAG (of the SSSP DAG of $s$) rooted in $w$ to the betweenness of $v$, whereas the term $1$ is the contribution of $w$ itself. For all nodes $v$ such that $\{w : v \in P_s(w)\} = \emptyset$ (i.e. the nodes that have no successors), we know that $\delta_{s \bullet}(v) = 0$. Starting from these nodes, we can compute $\delta_{s \bullet}(v)\ \forall v \in V$ by ``walking up'' the SSSP DAG rooted in $s$, using Eq.~(\ref{eq:dep}). Notice that it is fundamental that we process the nodes in order of decreasing distance from $s$, because to correctly compute $\delta_{s \bullet}(v)$, we need to know $\delta_{s \bullet}(w)$ for all successors of $v$. This can be done by inserting the nodes into a stack as soon as they are extracted from the SSSP (priority) queue in the first step. The betweenness of $v$ is then simply computed as $\sum_{s \neq v} \delta_{s \bullet}(v)$.

\section{Dynamic augmented APSP}
\label{sec:dyn-apsp}
As mentioned in Section~\ref{sec:brandes}, also dynamic algorithms based on \textsf{BA} build on its two steps. In the following, we will see how \textsf{KDB}~\cite{kourtellis2014scalable} and \textsf{KWCC}~\cite{DBLP:conf/asunam/KasWCC13} update the augmented APSP data structures (i.e. distances and number of shortest paths) after an edge insertion or a weight decrease. One difference between these two approaches is that \textsf{KDB} does not store the predecessors explicitly, whereas \textsf{KWCC} does. However, since in~\cite{kourtellis2014scalable} it was shown that keeping track of the predecessors only introduces overhead, we report a slightly-modified version of \textsf{KWCC} that recomputes them ``on the fly'' when needed (we will also use this version in our experiments in Section~\ref{sec:experiments}).
 We will then introduce our new approach in Section~\ref{sec:new-apsp}.
\subsection{Algorithm by Kourtellis et al. (\textsf{KDB})}
\label{sec:kdb1}
Let $(u, v)$ be the new edge inserted into $G$ (we recall that \textsf{KDB} works only on unweighted graphs, so edge weight modifications are not supported). For each source node $s \in V$, there are three possibilities: $(i)$ $d(s, u) = d(s, v)$, $(ii)$ $|d(s, u) - d(s, v)| = 1$ and $(iii)$ $|d(s, u) - d(s, v)| > 1$ (in case $(ii)$ and $(iii)$, let us assume that $d(s, u) < d(s, v)$ without loss of generality). We recall that $d$ is the distance \textit{before} the edge insertion.

In the first case, it is easy to see that the insertion does not affect any shortest path rooted in $s$, and therefore nothing needs to be updated for $s$. 

In case $(ii)$, the distance between $s$ and the other nodes is not affected, since there already
existed an alternative shortest-path from $s$ to $v$. However, the insertion creates new shortest paths from $s$ to to $v$ and consequently to all the nodes $t$ in the sub-DAG (of the SSSP DAG from $s$) rooted in $v$. To account for this, for each of these nodes $t$, we add $\sigma_{su} \cdot \sigma_{vt}$ to the old value of $\sigma_{st}$ (where $\sigma_{su} \cdot \sigma_{vt}$ is the number of new shortest paths between $s$ and $t$ going through $(u, v)$).

Finally, in case $(iii)$, a part of the sub-DAG rooted in $v$ might get closer to $s$. This case is handled with a BFS traversal rooted in $v$. In the traversal, all neighbors $y$ of nodes $x$ extracted from the BFS queue are examined and all the ones such that $d(s, y) \geq d'(s, x)$ are also enqueued. For each traversed node $y$, the new distance $d'(s, y)$ is computed as $\min_{z : (z, y) \in E} d'(s, z) + 1$ and the number of shortest paths $\sigma'_{sy}$ as $\sum_{z \in P'_s(y)} \sigma_{sz}$.

\subsection{Algorithm by Kas et al. (\textsf{KWCC})}
\label{sec:kwcc1}
\begin{wrapfigure}{r}{0.31\textwidth}
  \begin{center}
    \vspace{-5ex}
    \includegraphics[width=0.22\textwidth]{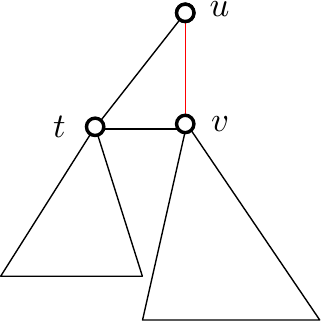}
  \end{center}
  \vspace{-4ex}
    \caption{Insertion of $(u,v)$.}
      \vspace{-2ex}
  \label{fig:comparison}
\end{wrapfigure}
\textsf{KWCC} updates the augmented APSP based on a dynamic APSP algorithm by Ramalingam and Reps~\cite{DBLP:journals/tcs/RamalingamR96}. Instead of checking for each source $s$ whether the new edge (or the weight decrease) changes the SSSP DAG rooted in $s$, \textsf{KWCC} first identifies the \textit{affected sources} $S = \{ s: d(s, v) \geq d(s, u) + \omega'(u, v)\}$. These are exactly the nodes for which there is some change in the SSSP DAG. The affected sources are identified by running a pruned BFS rooted in $u$ on $G$ transposed (i.e. the graph obtained by reversing the direction of edges in $G$). For each node $s$ traversed in the BFS, \textsf{KWCC} checks whether the neighbors of $s$ are also affected sources and, if not, it does not continue the traversal from them.
Notice that even on weighted graphs, a (pruned) BFS is sufficient since we already know all distances to $v$ and we can basically sidestep the use of a priority queue. 

Once all affected sources $s$ are identified, \textsf{KWCC} starts a pruned BFS rooted in $v$ for each of them. In the pruned BFS, only nodes $t$ such that $d(s,t) \geq d(s, u) + \omega'(u, v) + d(v, t)$ are traversed (the \textit{affected targets} of $s$). The new distance $d'(s, t)$ is set to $d(s, u) + \omega'(u, v) + d(v, t)$ and the new number of shortest paths $\sigma'(s, t)$ is set to $\sum_{z \in P'_s(t)} \sigma_{sz}$ as in \textsf{KDB}.
Compared to \textsf{KDB}, the augmented APSP update of \textsf{KWCC} requires fewer operations. First, it efficiently identifies the affected sources instead of checking all nodes. Second, in case $(iii)$, \textsf{KDB} might traverse more nodes than \textsf{KWCC}. For example, assume $(u, v)$ is a new edge and the resulting SSSP DAG of $u$ is as in Figure~\ref{fig:comparison}. Then, \textsf{KWCC} will prune the BFS in $t$, since $d(u, t) < d(u, v) + d(v, t)$, skipping all the SSSP DAGs rooted in $t$. On the contrary, \textsf{KDB} will traverse the whole subtree rooted in $t$, although neither the distances nor the number of shortest paths from $u$ to those nodes are affected. The reason for this will be made clearer in Section~\ref{sec:kdb2}.

\subsection{Faster augmented APSP update}
\label{sec:new-apsp}
\begin{wrapfigure}{r}{0.3\textwidth}
  \begin{center}
    \vspace{-4ex}
    \includegraphics[width=0.20\textwidth]{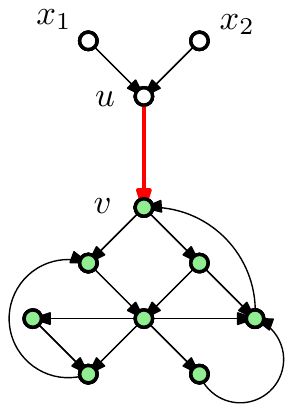}
  \end{center}
  \vspace{-5ex}
  \caption{Affected targets (in green) and affected sources ($x_1, x_2, u$).}
  \label{fig:new-edge}
\end{wrapfigure}
To explain our idea for improving the APSP update step, let us start with an example, shown
in Figure~\ref{fig:new-edge}. 
The insertion of $(u, v)$ decreases the distance from nodes $x_1, x_2, u$ to 
all the nodes shown in green. 
\textsf{KWCC} would first identify the affected sources $S = \{x_1, x_2, u\}$ and, for each of them, run a pruned BFS rooted in $v$. 
This means we are repeating almost exactly the same procedure for each of the affected sources. We clearly have to update the distances and number of shortest paths between each affected source and the affected targets (and this cannot be avoided). However, \textsf{KWCC} also goes through the outgoing edges of each affected target multiple times, leading to a worst-case running time of $O(mn)$.\footnote{Notice that this is true also for \textsf{KDB}, with the difference that \textsf{KDB} starts a BFS from each node instead of first identifying the affected sources and that it also visits additional nodes.}
Our basic idea is to avoid this redundancy and is based on
the following proposition (a similar result was proven also in~\cite{Lin:1991:DSP:105582.105591}).
\begin{proposition}
\label{lem:source-nodes}
Let $t \in V$ and $y \in P_v(t)$ be given. Then, $S(t) \subseteq S(y)$. 
\end{proposition}
\begin{proof}
Let $s$ be any node in $S(t)$, i.e. either $d'(s,t) = d(s,t)$ and $\sigma'_{st} \neq \sigma_{st}$ (case $(i)$), or $d'(s,t) < d(s,t)$ (case $(ii)$). 
We want to show that $s \in S(y)$.

 Before proving this, we show that $y$ has to be in $P'_s(t)$. In fact, if $s \in S(t)$, there have to be shortest paths between $s$ and $t$ going through $(u, v)$, i.e. $d'(s, t) = d(s, u) + \omega'(u, v) + d(v, t)$. On the other hand, we know $y \in P_v(t)$ and thus 
\begin{equation}
\label{eq:proof}
d'(s, t) = d(s, u) + \omega'(u, v) + d(v, y) + \omega(y, t).
\end{equation}
Now, $d(s, u) + \omega'(u, v) + d(v, y)$ cannot be larger than $d'(s, y)$, or this would mean that $d'(s, t) > d'(s, v) + \omega(y, t)$, which contradicts the triangle inequality. Also, $d(s, u) + \omega'(u, v) + d(v, y)$ cannot be smaller than $d'(s, y)$ by definition of distance. Thus, $d'(s, y) = d(s, u) + \omega'(u, v) + d(v, y)$. If we substitute this in Eq.~(\ref{eq:proof}), we obtain $d'(s, t) = d'(s, y) + \omega(y, t)$, which means $y \in P'_s(t)$.

Now, let us consider case $(i)$. We have two options: either $y$ was a predecessor of $t$ from $s$ also before the edge update, i.e. $y \in P_s(t)$, or it was not. If it was not, it means $d(s, y) + \omega(y, t) > d(s, t) = d'(s, t) = d'(s, y) + \omega(y, t)$, which implies $d(s, y) > d'(s, y)$ and thus $s \in S(y)$.
If it was, we can similarly show that $d(s, y) = d'(s, y)$. Since we have seen before that $ d'(s, y) = d(s, u) + \omega'(u, v) + d(v, y)$, there has to be at least one new shortest path from $s$ to $y$ in $G'$ going through $(u, v)$, which means $\sigma'_{sy} > \sigma_{sy}$ and therefore $s \in S(y)$.

Case $(ii)$ can be easily proven by contradiction. We know $d(s,t) \leq d(s,y) + \omega(y,t)$ (by the triangle inequality) and that $\omega'(y,t) = \omega(y,t)$. Thus, if it were true that $d(s,y) = d'(s,y) $ then
\begin{equation}
d(s,t) \leq d(s,y) + \omega(y,t) = d'(s,y) + \omega(y,t) = d'(s,t),
\end{equation}
which contradicts our hypothesis that $d'(s,t) < d(s,t)$ (case $(ii)$). Thus, $d(s,y) \neq d'(s,y) $. Since pairwise distances in $G'$ can only be equal to or shorter than pairwise distances in $G$, $d(s,y) \neq d'(s,y)$ implies $d(s,y) > d'(s,y) $ and thus $s \in S(y)$.

\end{proof}

In particular, this implies that $S(t) \subseteq S(v)$ for each $t \in T(u)$. Consequently, it is sufficient to compute $S(v)$ and $T(u)$ once via two pruned BFSs.
Our approach is described in Algorithm~\ref{edge_insertion}. The pruned BFS to compute $S(v)$ is performed in Line~\ref{aff-sources}. Then, a pruned BFS from $v$ is executed, whereby for each $t \in T(u)$ we store one of its predecessors $p(t)$ in the BFS (Line~\ref{predecessor}). 

Let $d^{\star}(s, t)$ be the length of a shortest path between $s$ and $t$ going through $(u, v)$, i.e. $d^{\star}(s, t) := d(s, u) + \omega'(u, v) + d(v, t)$. To finally compute $S(t)$ all that is left to do is to test whether $d^{\star}(s, t) \leq d(s, t)$ for each $s \in S(p(t))$ once we remove $t$ from the queue (Lines~\ref{loop-pred-begin} - \ref{loop-pred-end}). 
Note that this implies that $S(p(t))$ was already computed. 
In case $d^{\star}(s,t) < d(s,t)$, the path from $s$ to $t$ via edge $(u,v)$ is shorter than before and therefore we set $d'(s,t)$ to $d^{\star}(s,t)$ and $\sigma'_{st}$ to $\sigma_{su} \cdot \sigma_{vt}$, since all new shortest paths now go through $(u,v)$). 
Also in case of equality ($d^{\star}(s,t) = d(s,t)$), $s$ is in $S(t)$, since its number of shortest paths has changed. Consequently we set $\sigma'_{st}$ to $\sigma_{st} + \sigma_{su} \cdot \sigma_{vt}$ (since in this case also old shortest paths are still valid). 
If $d^{\star}(s, t) > d(s, t)$, the edge $(u, v)$ does not lie on any shortest path from $s$ to $t$, hence $s \notin S(t)$ (and $s$ is not added to $S(t)$ in Lines~\ref{line:add-sources-begin} - \ref{line:add-sources-end}). 
 \IncMargin{1em}
\begin{algorithm}[tb]
\begin{footnotesize}
\SetKwInOut{Input}{Input}\SetKwInOut{Output}{Output}
\SetKwInOut{Assume}{Assume}
\SetKwFunction{insert}{insert}
\SetKwFunction{findAffectedSources}{findAffectedSources}
\SetKwData{visited}{visited}
\SetKwData{Pred}{$P_v$}
\SetKwData{source}{$S$}
\Input{Graph $G=(V, E)$, edge insertion/weight decrease $(u,v,\omega'(u,v))$, $d(s, t)$, $\sigma_{st},\  \forall (s, t) \in V^2$}
\Output{Updated $d'(s, t)$, $\sigma '_{st},\  \forall (s, t) \in V^2$}
\Assume{Initially $d'(s, t) = d(s, t)$ and $\sigma'_{st} = \sigma_{st} \ \ \forall (s, t) \in V^2$}
$vis(v) \leftarrow \text{false}\ \forall v \in V$\;
\If {$\omega'(u,v) \leq d(u,v)$} { 
    $S(v) \leftarrow$ \findAffectedSources{$G, (u,v,\omega'(u,v))$}\; \label{aff-sources}
  $d(u,v) \leftarrow \omega'(u, v)$\; \label{up-dist}
  $Q \leftarrow \emptyset$\;
  $p(v) \leftarrow v$\;
  $Q.push(v)$\;
  $vis(v) \leftarrow \text{true}$\;
  \While {$Q.length() > 0$}{ 
   \label{phase2-begin}
    $t = Q.front()$\; \label{dequeue}
    \ForEach {$s \in \source(p(t))$} { 
    \label{loop-pred-begin}
      \If {$d(s,t) \geq d(s,u) + \omega'(u,v)+ d(v,t)$} { 
      \label{line:new-dist}
      \If {$d(s,t) > d(s,u) + \omega'(u,v)+ d(v,t)$} { 
        		$d '(s,t) \leftarrow d(s,u) + \omega'(u,v)+ d(v,t)$\; \label{line:up-dist}
		$\sigma '_{st} \leftarrow 0$\;
	}
	$\sigma '_{st} \leftarrow \sigma '_{st} + \sigma_{su} \cdot \sigma_{vt}$\;
        \If{$t \neq v$} {\label{line:add-sources-begin}
        $S(t).insert(s)$\;
        }\label{line:add-sources-end}
    } 
    }
    \label{loop-pred-end}
    \ForEach {$w$ s.t. $(t, w) \in E$} { 
    \label{out-neighbors-begin}
      \If {not $vis(w)$ and $d(u, w) \geq \omega'(u, v) + d(v, w)$}{ 
      \label{out-neighbors-if} 		
        $Q.push(w)$\; \label{enqueue}
          $vis(w) \leftarrow \text{true}$\;
          $p(w) \leftarrow t$\; \label{predecessor}
      } 
      \label{out-neighbors-end}
    } 
  
  } 
  \label{phase2-end}
} 
\end{footnotesize}
\caption{Augmented APSP update}
\label{edge_insertion}
\end{algorithm}
\DecMargin{1em}

\section{Dynamic dependency accumulation}
\label{sec:dependency}
After updating distances and number of shortest paths, dynamic algorithms need to update the betweenness scores. This means increasing the score of all nodes that lie in new shortest paths, but also decreasing that of nodes that used to be in old shortest paths between affected nodes. Again, we will first see how \textsf{KDB} and \textsf{KWCC} update the dependencies and then we will present our new approach in Section~\ref{sec:new-dep}.

\subsection{Algorithm by Kourtellis et al. (\textsf{KDB})}
\label{sec:kdb2}
In addition to $d$ and $\sigma$, \textsf{KDB} keeps track of the old dependencies $\delta_{s \bullet}(v)\ \forall s, v \in V$. The dependency update is done in a way similar to \textsf{BA} (see Section~\ref{sec:brandes}). Also in this case, nodes $v$ are processed in decreasing order of their new distance $d'(s, v)$ from $s$ (otherwise it would not be possible to apply Eq.~(\ref{eq:dep})). However, in this case we would only like to process nodes for which the dependency has actually changed. To do this, while still making sure that the nodes are processed in the right order, \textsf{KDB} replaces the stack used in \textsf{BA} with a bucket list. Every node that is traversed during the APSP update is inserted into the bucket list in a position equal to its new distance from s. Then, nodes are extracted from the bucket list starting from the ones with maximum distance. Every time a node $v$ is extracted, we compute its new dependency as $\delta'_{s \bullet}(v) = \sum_{w : v \in P'_s(w)} \frac{\sigma'_{sv}}{\sigma'_{sw}} (1 + \delta'_{s \bullet}(w))$. Since we are processing the nodes in order of decreasing new distance, we can be sure that $\delta'_{s \bullet}(v)$ is computed correctly.
The score of $v$ is then updated by adding the new dependency $\delta'_{s \bullet}(v)$ and subtracting the old $\delta_{s \bullet}(v)$, which was previously stored.
Also, all neighbors $y \in P'_{s}(v)$ that are not in the bucket list yet are inserted at level $d'(s, y) = d'(s, v) - 1$. Notice that, in the example in Figure~\ref{fig:comparison}, all the nodes in the sub-DAG of $t$ are necessary to compute the new dependency of $t$, although they have not been affected by the insertion. This is why they are traversed during the APSP update.

\subsection{Algorithm by Kas et al. (\textsf{KWCC})}
\label{sec:kwcc2}
\textsf{KWCC} does not store dependencies. On the contrary, for every node pair $(s, t)$ for which either $d(s, t)$ or $\sigma_{st}$ has been affected by the insertion, all the nodes in the new shortest paths and the ones in the old shortest paths between $s$ and $t$ are processed.  More specifically, starting from $t$, all the nodes $y \in P'_s(t)$ are inserted into a queue. When a node $y$ is extracted, we increase its betweenness by $\sigma'(s, y) \cdot \sigma' (y, t) / \sigma'(s, t)$ (i.e. the fraction of shortest paths between $s$ and $t$ going through $y$). Then, also $y$ enqueues all nodes in $P'_s(y)$ and the process is repeated until we reach $s$. Decreasing the betweenness of nodes in the old paths is done in a similar fashion, with the only difference that nodes in $P_s(y)$ are enqueued (instead of nodes in $P'_s(y)$) and that $\sigma(s, y) \cdot \sigma (y, t) / \sigma(s, t)$ is subtracted from the scores of processed nodes. Notice that the worst-case complexity of this approach is $O(n^3)$, whereas that of \textsf{KDB} is $O(nm)$. This cubic running time is due to the fact that, for each affected node pair $(s, t)$ (at most $\Theta(n^2))$, there could be up to $\Theta(n)$ nodes lying in either one of the old or new shortest paths between $s$ and $t$. (In the running time analysis of~\cite{kourtellis2014scalable}, this is represented by the term $|\sigma_{old}| I$.) This means that, if many nodes are affected, \textsf{KWCC} can even be slower than recomputation with \textsf{BA}. On the other hand, we have seen in Section~\ref{sec:kwcc1} that \textsf{KDB} also processes nodes for which the betweenness has not changed (see Figure~\ref{fig:comparison} and its explaination), which in some cases might result in a higher running time than \textsf{KWCC}. 

\subsection{Faster betweenness update}
\label{sec:new-dep}
We propose a new approach for updating the betweenness scores. As \textsf{KWCC}, we do not store the old dependencies (resulting in a lower memory requirement) and we only process the nodes whose betweenness has actually been affected. However, we do this by accumulating contributions of nodes only once for each affected source, in a fashion similar to \textsf{KDB}. 
For an affected source $s \in S$ and for any node $v \in V$, let us define $\Delta_{s, \bullet}(v)$ as $\sum_{t \in T(s)} \sigma_{st}(v)/\sigma_{st}$. This is the contribution of nodes whose old shortest paths from $s$ went through $v$, but which have been affected by the edge insertion. Analogously, we can define $\Delta'_{s, \bullet}(v)$ as $\sum_{t \in T(s)} \sigma'_{st}(v)/\sigma'_{st}$. Then, the new dependency $\delta'_{s, \bullet} (v)$ can be expressed as:
\begin{equation}
\delta'_{s, \bullet} (v) = \delta_{s, \bullet} (v) - \Delta_{s, \bullet}(v) + \Delta'_{s, \bullet}(v)
\end{equation}
Notice that for all nodes $t \notin T(s)$, $\sigma'_{st} = \sigma_{st}$ and $\sigma'_{st}(v) = \sigma_{st}(v)$, therefore their contribution to $\delta_{s, \bullet} (v)$ is not affected by the edge update.
The new betweenness $c'_B(v)$ can then be computed as $c_B(v) - \sum_{s \in S} \Delta_{s, \bullet}(v) +  \sum_{s \in S} \Delta'_{s, \bullet}(v)$.
The following theorem allows us to compute $\Delta_{s, \bullet}(v)$ and $\Delta'_{s, \bullet}(v)$ efficiently.
 \begin{theorem}
\label{theo}
For any $s \in T, v \in V$:
\[
  \Delta_{s, \bullet}(v)= \sum_{w : v \in P_s(w) \wedge w \in T(s)} \sigma_{sv}/\sigma_{sw} (1 + \Delta_{s, \bullet}(w)) + \sum_{w : v \in P_s(w) \wedge w \notin T(s)} \sigma_{sv}/\sigma_{sw} \cdot \Delta_{s, \bullet}(w)\ .
\]
Similarly:
\[
  \Delta'_{s, \bullet}(v)= \sum_{w : v \in P'_s(w) \wedge w \in T(s)} \sigma'_{sv}/\sigma'_{sw} (1 + \Delta'_{s, \bullet}(w)) + \sum_{w : v \in P'_s(w) \wedge w \notin T(s)} \sigma'_{sv}/\sigma'_{sw} \cdot \Delta'_{s, \bullet}(w)\ .
\]
\end{theorem}
 \begin{proof}
 We prove only the equation for $\Delta_{s, \bullet}(v)$, the one for $\Delta'_{s, \bullet}(v)$ can be proven analogously. 
 Let t be any node in $T(s)$, $t \neq v$. Then, $\sigma_{st}(v)/\sigma_{st}$ can be rewritten as $\sum_{w : v \in P_s(w)} \sigma_{st}(v, w)/\sigma_{st}$, where $\sigma_{st}(v, w)$ is the number of shortest paths between $s$ and $t$ going through both $v$ and $w$.
Then:
 \[
 \Delta_{s, \bullet}(v) = \sum_{t \in T(s)} \sigma_{st}(v)/\sigma_{st} = \sum_{t \in T(s)} \sum_{w : v \in P_s(w)} \sigma_{st}(v, w)/\sigma_{st} = \sum_{w : v \in P_s(w)} \sum_{t \in T(s)} \sigma_{st}(v, w)/\sigma_{st}\ .
 \]
 Now, of the $\sigma_{sw}$ paths from $s$ to $w$, there are $\sigma_{sv}$ many that also go through $v$. Therefore, for $t \neq w$, there are $\frac{\sigma_{sv}}{\sigma_{sw}} \cdot \sigma_{st}(w)$ shortest paths from $s$ to $t$ containing both $v$ and $w$, i.e. $\sigma_{st}(v, w) = \frac{\sigma_{sv}}{\sigma_{sw}} \cdot \sigma_{st}(w)$. On the other hand, if $t = w$, $\sigma_{st}(v, w)$ is simply $\sigma_{sv}$. Therefore, we can rewrite the equation above as:
 \[
 \begin{split}
%
&  \sum_{w : v \in P_s(w) \wedge w \in T(s)} \left\{ \frac{\sigma_{sv}}{\sigma_{sw}} + \sum_{t \in T(s) - \{w\}}  \frac{\sigma_{st}(v, w)}{\sigma_{st}} \right\}
 + \sum_{w : v \in P_s(w) \wedge w \notin T(s)} \sum_{t \in T(s)} \frac{\sigma_{st}(v, w)}{\sigma_{st}}\\
  & = \sum_{w : v \in P_s(w) \wedge w \in T(s)} \frac{\sigma_{sv}}{\sigma_{sw}} \left(1 + \sum_{t \in T(s) - \{w\}}  \frac{\sigma_{st}(w)}{\sigma_{st}} \right)
 + \sum_{w : v \in P_s(w) \wedge w \notin T(s)} \frac{\sigma_{sv}}{\sigma_{sw}} \sum_{t \in T(s)} \frac{\sigma_{st}(w)}{\sigma_{st}}\\
& = \sum_{w : v \in P_s(w) \wedge w \in T(s)} \frac{\sigma_{sv}}{\sigma_{sw}} (1 + \Delta_{s, \bullet}(w)) + \sum_{w : v \in P_s(w) \wedge w \notin T(s)} \frac{\sigma_{sv}}{\sigma_{sw}} \cdot \Delta_{s, \bullet}(w)\ .
 \end{split}
 \]
 \end{proof}
Theorem~\ref{theo} allows us to accumulate the dependency changes in a way similar to \textsf{BA}. To compute $ \Delta_{s, \bullet}$, we need to process nodes in decreasing order of $d(s, \cdot)$, whereas to compute $\Delta'_{s, \bullet}$ we need to process them in decreasing order of $d'(s, \cdot)$. To do this, we use two priority queues $PQ_s$ and $PQ'_s$ (if the graph is unweighted, we can use bucket lists as the ones used in \textsf{KDB}). Notice that nodes $w$ such that $\sigma_{st}(w) = 0 \wedge \sigma'_{st}(w) = 0 \  \forall t \in T(s)$ do not need to be added to the queue. $PQ_s$ and $PQ'_s$ are filled with all nodes in $T(s)$ during the APSP update in Algorithm~\ref{edge_insertion}. In $PQ_s$, nodes $w$ are inserted with priority $d(s, w)$ and $PQ'_s$ with priority $d'(s, w)$. Algorithm~\ref{algo:dyn_bc} shows how we decrease betweenness of nodes that lied in old shortest paths from $s$ (notice that this is repeated for each $s \in S(v)$). In Lines~\ref{line:theo1} - \ref{line:theo2}, Theorem~\ref{theo} is applied to compute $\Delta_{s, \bullet}(y)$ for each predecessor $y$ of $w$. Then, $y$ is also enqueued and this is repeated until $PQ_s$ is empty (i.e. when we reach $s$). The betweenness update of nodes in the new shortest paths works in a very similar way. The only difference is that $PQ'_s$ is used instead of $PQ$, that $d'$ and $\sigma '$ are used instead of $d$ and $\sigma$ and that $\Delta'_{s, \bullet}$ is added to $c_B$ and not subtracted in Line~\ref{dec-betw}. At the end of the update, $\sigma$ is set to $\sigma'$ and $d$ is set to $d'$.

In undirected graphs, we can notice that $\sum_{s \in S(w)} \Delta_{s, \bullet}(w) = \sum_{t \in T(w)} \Delta_{t, \bullet}(w)$. Thus, to account also for the changes in the shortest paths between $w$ and the nodes in $T(w)$, $2 \Delta_{s, \bullet}$ is subtracted from $c_B(w)$ in Line~\ref{dec-betw} (and analogously $2 \Delta'_{s, \bullet}$ is added in the update of nodes in the new shortest paths).

\begin{algorithm}[t]
 \begin{footnotesize}
\LinesNumbered
$\Delta_{s, \bullet}(u) \leftarrow 0\ \forall u \in V$\;
\While{$PQ_s \neq \emptyset$}{
	$w \leftarrow PQ_s$.extractMax()\; \label{extractMax}
	$c_B(w) \leftarrow c_B(w) - \Delta_{s, \bullet}(w)$\; \label{dec-betw}
	\ForEach {$y$ s.t. $(y, w) \in E$}{ \label{bc-loop-begin}
		\If{$y \neq s$ and $d(s, w) = d(s, y) + \omega(y, w)$}{
			\If{$w \in T(s)$}{ \label{line:theo1}
				$c \leftarrow \frac{\sigma_{sy}}{\sigma_{sw}} \cdot (1 + \Delta_{s, \bullet}(w))$\;
			} \Else {
				$c \leftarrow \frac{\sigma_{sy}}{\sigma_{sw}} \cdot \Delta_{s, \bullet}(w)$\;
			}  \label{line:theo2}
			\If{$ y \notin PQ_s$} {
			\text{Insert $y$ into $PQ_s$ with priority $d(s, y)$}\; \label{bc-enqueue}
			}
			$\Delta_{s, \bullet}(y) \leftarrow \Delta_{s, \bullet}(y) + c$\;
		}
	} \label{bc-loop-end}
}
\end{footnotesize}
\caption{Betweenness update for nodes in old shortest paths}
\label{algo:dyn_bc}
\end{algorithm}

\section{Time complexity}
\label{sec:complexity}
Let us study the complexity of our two new algorithms for updating APSP and betweenness scores described in Section~\ref{sec:new-apsp} and Section~\ref{sec:new-dep}, respectively. 
We define the \textit{extended size} $||A||$ of a set of nodes $A$ as the sum of the number of nodes in $A$ and the number of edges that have a node of $A$ as their endpoint. Then, the following holds.
\begin{theorem}
\label{theo_compl1}
The running time of Algorithm~\ref{edge_insertion} for updating the augmented APSP after an edge insertion (or weight decrease) $(u, v, \omega'(u,v))$ is $\Theta(||S(v)|| + ||T(u)|| + \sum_{y \in T(u)} |S(p(y))|)$, where $p(y)$ can be any node in $P_u(y)$.
\end{theorem}
\begin{proof}
The function \texttt{findAffectedSources} in Line~\ref{aff-sources} identifies the set of affected sources starting a BFS in $v$ and visiting only the nodes $s \in S(v)$. This takes $\Theta(||S(v)||)$, since this pruned BFS visits all nodes in $S(v)$ and their incident edges.
Then, the while loop of Lines~\ref{phase2-begin} -~\ref{phase2-end} identifies all the affected targets $T(u)$ with a pruned BFS. This part (excluding Lines~\ref{loop-pred-begin} -~\ref{loop-pred-end}) requires $\Theta(||T(u)||)$ operations, since all affected targets and their incident edges are visited. 
In Lines~\ref{loop-pred-begin} -~\ref{loop-pred-end}, for each affected node $t \in T(u)$, all the affected sources of the predecessor $p(y)$ of $y$ are scanned. This part requires in total $\Theta(\sum_{t \in T(u)} |S(p(y))|)$ operations.
\end{proof}
Notice that, since $|S(p(y))|$ is $O(n)$ and both $||T(u)||$ and $||S(v)||$ are $O(n+m)$, the worst-case complexity of Algorithm~\ref{edge_insertion} is $O(n^2)$. 
To show the complexity of the dependency update described in Algorithm~\ref{algo:dyn_bc}, let us introduce, for a given source node $s$, the set $\tau(s) := T(s) \cup \{w \in V : \Delta_{s, \bullet}(w) > 0 \}$. Then, the following theorem holds.
\begin{theorem}
\label{theo_compl2}
The running time of Algorithm~\ref{algo:dyn_bc} is $\Theta(|| \tau(s) || + |\tau(s)| \log |\tau(s)|)$ for weighted graphs and $\Theta(|| \tau(s) ||)$ for unweighted graphs. 
\end{theorem}
\begin{proof}
In the following, we assume a binary heap priority queue for weighted graphs and a bucket list priority queue for unweighted graphs. 
Then, the \texttt{extractMax()} operation in Line~\ref{extractMax} requires constant time for unweighted and logarithmic time for weighted graphs. Also, for each node extracted from $PQ$, all neighbors are visited in Lines~\ref{bc-loop-begin} - \ref{bc-loop-end}. Therefore, it is sufficient to prove that the set of nodes inserted into (and therefore extracted from) $PQ$ is exactly $\tau(s)$.
As we said in the description of Algorithm~\ref{algo:dyn_bc}, $PQ$ is initially populated with the nodes in $T(s)$. Then, all nodes $y$ inserted into $PQ$ in Line~\ref{bc-enqueue} are nodes that lied in at least one shortest path between $s$ and a node in $T(s)$ before the insertion. This means that there is at least one $t \in T(s)$ such that $\sigma_{st}(y) > 0$, which implies that $\Delta_{s, \bullet}(y) > 0$, by definition of $\Delta_{s, \bullet}(y)$.
\end{proof}
The running time necessary to increase the betweenness score of nodes such that $\Delta'_{s, \bullet} > 0$ can be computed analogously, defining $\tau'(s) = T(s) \cup \{w \in V : \Delta'_{s, \bullet}(w) > 0 \}$. Overall, the running time of the betweenness update score described in Section~\ref{sec:new-dep} is $\Theta(\sum_{s \in S} ||\tau(s) || + ||\tau'(s)||)$ for unweighted and $\Theta(\sum_{s \in S} ||\tau(s) || + ||\tau'(s)|| +  |\tau(s)| \log |\tau(s)| +  |\tau '(s)| \log |\tau '(s)|)$ for weighted graphs. Consequently, in the worst case, this is $O(nm)$ for unweighted and $O(n(m + n \log n))$ for weighted graphs, which matches the running time of \textsf{BA}. For sparse graphs, this is asymptotically faster than \textsf{KWCC}, which requires $\Theta(n^3)$ operations in the worst case.

\section{Experimental Results}
\label{sec:experiments}
\paragraph*{Implementation and settings} 
For our experiments, we implemented \textsf{BA}, \textsf{KDB}, \textsf{KWCC}, and our new approach, which we refer to as \textsf{iBet} (from Incremental Betweenness). All the algorithms were implemented in C++, building on the open-source \textit{NetworKit}
framework~\cite{DBLP:journals/corr/StaudtSM14}. 
All codes are sequential; they were executed on a 64bit machine with 2 x 8 Intel(R) Xeon(R) E5-2680 cores at 2.7 GHz with 256 GB RAM with a single thread on a single CPU. 
\paragraph*{Data sets and experimental design} For our experiments, we consider a set of real-world networks belonging to different domains, taken from SNAP~\cite{snapnets}, KONECT~\cite{DBLP:conf/www/Kunegis13}, and LASAGNE (\url{piluc.dsi.unifi.it/lasagne
}). Since \textsf{KDB} cannot handle weighted graphs and the pseudocode given in~\cite{kourtellis2014scalable} is only for undirected graphs, all graphs used in the experiments are undirected and unweighted. 
The networks are reported in Table~\ref{table:edge_undirected}. Due to the time required by the static algorithm and the memory constraints of all dynamic algorithms ($\Theta(n^2)$), we only considered networks with up to about 26000 nodes.

To simulate real edge insertions, we remove an existing edge from the graph (chosen uniformly at random), compute betweenness on the graph without the edge and then re-insert the edge, updating betweenness with the incremental algorithms (and recomputing it with \textsf{BA}). 
For all networks, we consider 100 edge insertions and report the average over these 100 runs.

\paragraph*{Experimental results} 
\begin{table}[h!]
   \caption{The table shows the average time taken by the static algorithm \textsf{BA} and the average speedups on \textsf{BA} of the incremental algorithms (geometric means).  The best result of each row is shown in bold font.}
\begin{center}
\begin{small}
 \vspace{-1ex}

  \begin{tabular}{ | l | r | r | l | r | r | r | r |}
    \cline{6-8}
\multicolumn{5}{c|}{} & \multicolumn{3}{c|}{Speedup on \textsf{BA}} \\
\hline
Graph	&	Nodes	&	Edges	& Type &	 \textsf{BA} [s]	&	\textsf{iBet}	&	\textsf{KDB}	&	\textsf{KWCC}\\
\hline
\texttt{HC-BIOGRID}	&	4 039	&	10 321		& bio. network  &	 6.06  & \textbf{77.87}   &  10.91  &  18.33\\   
\texttt{Mus-musculus}	&	4 610	&	5 747	& bio. network	&	3.32  &  \textbf{119.23}   &  9.40  &  11.21\\
\texttt{Caenor-elegans}	&	4 723	&	9 842  	&  metabolic	&  5.12  &   \textbf{130.89}   &  9.58  &  23.64\\
\texttt{ca-GrQc}	&	5 241	&	14 484	                 & coauthorship &	4.19  &  \textbf{206.55}   &  7.53  &  14.28\\
\texttt{advogato}	&	7 418	&	42 892		& social	&	14.65  &  \textbf{295.39}   &  27.69  &  18.45\\
\texttt{hprd-pp}	&	9 465	&	37 039			&  bio. network     &	30.29  &   \textbf{304.24}   &  11.33  &  45.90\\
\texttt{ca-HepTh}	&	9 877	&	25 973		& coauthorship	&	21.06  &   \textbf{199.04}   &  8.24  &  34.03\\
\texttt{dr-melanogaster}	&	10 625	&	40 781	& bio. network	&	40.76  &  \textbf{235.54}   &  7.94  &  48.57\\
\texttt{oregon1-010526}	&	11 174	&	23 409	& aut. systems	&	24.43  & \textbf{237.47}   &  15.20  &  21.64\\
\texttt{oregon2-010526}	&	11 461	&	32 730	& aut. systems	&	 30.07  &  \textbf{113.10}   &  17.23  &  23.08\\
\texttt{Homo-sapiens}	&	13 690	&	61 130	& bio. network	&	68.58  &   \textbf{237.61}   &  10.29  &  58.67\\
\texttt{GoogleNw}	&	15 763	&	148 585		& hyperlinks	&	90.42  &  \textbf{577.49}   &  90.01  &  33.80\\
\texttt{dip20090126}	&	19 928	&	41 202		& bio. network	&	115.56  & \textbf{51.54}   &  5.38  &  5.73\\
\texttt{as-caida20071105}	&	26 475	&	53 381	& aut. systems 	&	154.36  &   \textbf{173.90}   &  18.66  &  19.65\\ \hline
Geometric mean	&		&		&   &	  &   \textbf{179.1}   &  13.0  &  22.9\\
\hline
  \end{tabular}
  \end{small}
\end{center}
  \label{table:edge_undirected}
\end{table}
In Table~\ref{table:edge_undirected} the running times of \textsf{BA} for each graph and the speedups of the three incremental algorithms on \textsf{BA} are reported. The last line shows the geometric mean of the speedups on \textsf{BA} over all tested networks.
Our new method \textsf{iBet} clearly outperforms the other two approaches and is always faster than both of them. On average, \textsf{iBet} is faster than \textsf{BA} by a factor 179.1, whereas \textsf{KDB} by a factor 13.0 and \textsf{KWCC} by a factor 22.9.

Figure~\ref{fig:dep} compares the APSP update (on the left) and dependency update (on the right) steps for the \texttt{oregon1-010526} graph (a similar behavior was observed also for the other graphs of Table~\ref{table:edge_undirected}. 
On the left, the running time of the APSP update phase of the three incremental algorithms on 100 edge insertions are reported, sorted by the running time taken by \textsf{KDB}. It is clear that the APSP update of \textsf{iBet} is always faster than the competitors. This is due to the fact that \textsf{iBet} processes the edges between the affected targets only once instead of doing it once for each affected source as both \textsf{KDB} and \textsf{KWCC}. Also, the running time of the APSP update of \textsf{KDB} varies significantly. On about one third of the updates, it is basically as fast as \textsf{KWCC}. This means that in these cases, \textsf{KDB} only visits a small amount of nodes in addition to the affected ones (see Figure~\ref{fig:comparison} and its explanation). However, in other cases \textsf{KDB} can be much slower, as shown in the figure.

On the right of Figure~\ref{fig:dep}, the running times of the dependency update step are reported. Also for this step, \textsf{iBet} is faster than both \textsf{KDB} and \textsf{KWCC}. However, for this part there is not a clear winner between \textsf{KWCC} and \textsf{KDB}. In fact, in some cases \textsf{KDB} needs to process additional nodes in order to recompute dependencies, whereas \textsf{KWCC} only processes nodes in the shortest paths between affected nodes. However, \textsf{KDB} processes each node at most once for each source node $s$, whereas  \textsf{KWCC} might process the same node several times if it lies in several shortest paths between $s$ and other nodes (we recall that the worst-case running time of \textsf{KWCC} is $O(n^3)$, whereas that of \textsf{KDB} is $O(n m)$). Notice also that in some rare cases \textsf{KDB} is slightly faster than \textsf{iBet} in the dependency update. This is probably due to the fact that our implementation of \textsf{iBet} is based on a priority queue, whereas \textsf{KDB} on a bucket list.

\begin{figure}[bt]
\begin{center}
\includegraphics[width = 0.48\textwidth]{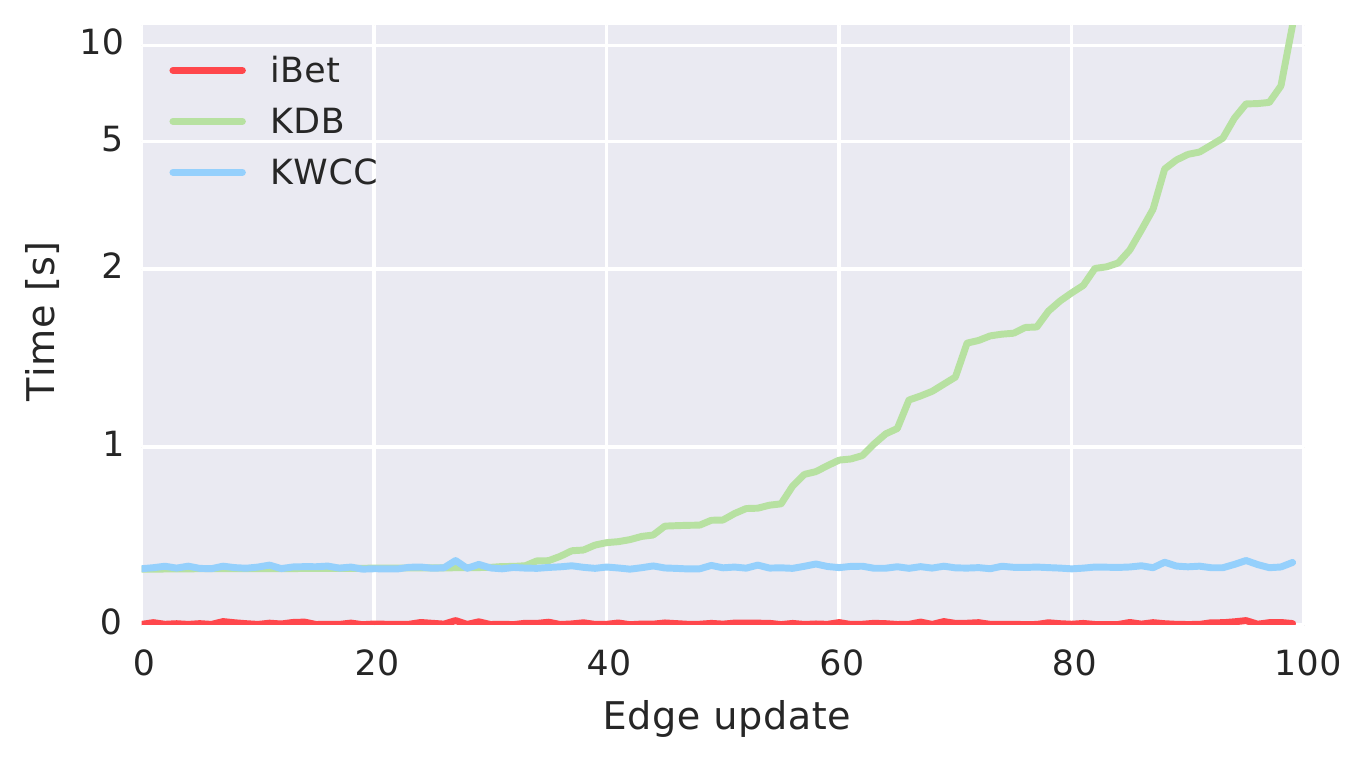}
\includegraphics[width = 0.48\textwidth]{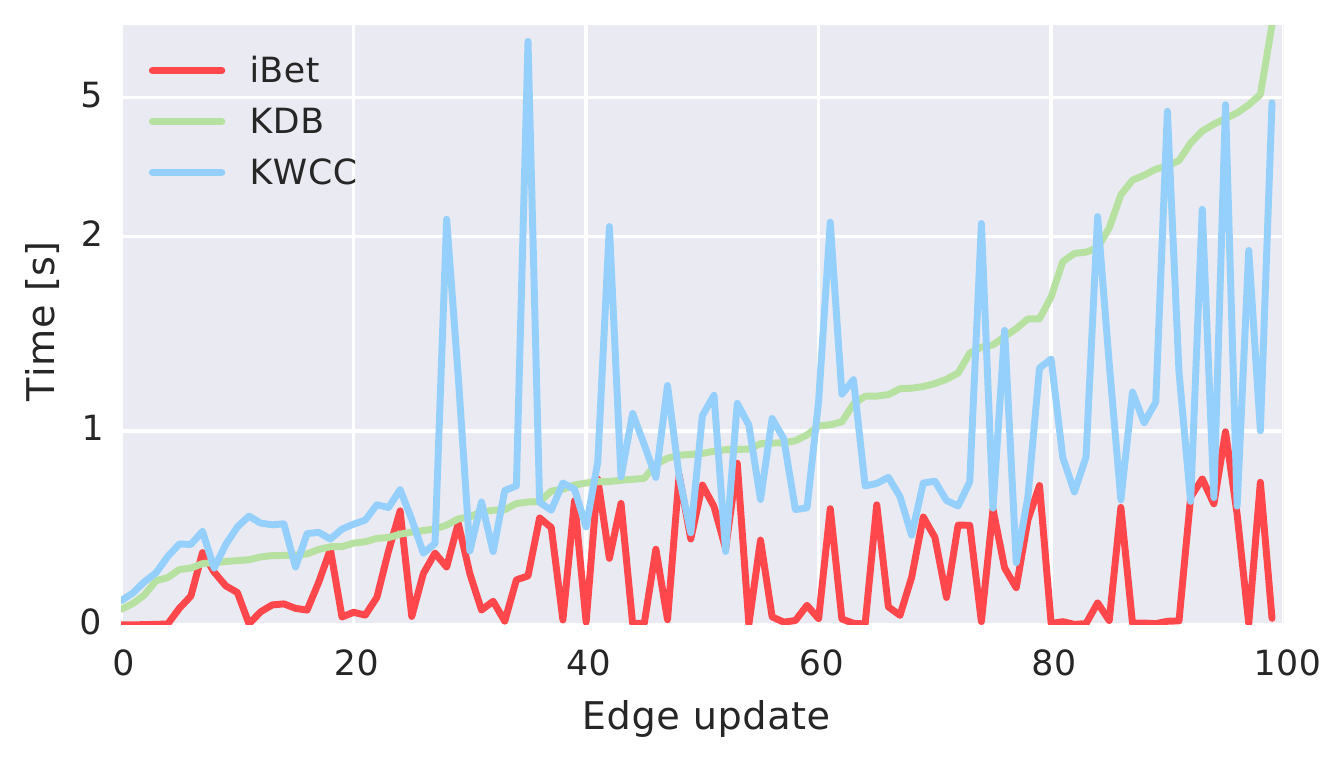}
\caption{Running times of \textsf{iBet}, \textsf{KDB} and \textsf{KWCC} for 100 edge updates on \texttt{oregon1-010526}. Left: times for the APSP update step. Right: times for the dependency update step.}
\label{fig:dep}
\end{center}
\vspace{-3ex}
\end{figure}

\begin{figure}[tb]
\begin{center}
\includegraphics[width = 0.48\textwidth]{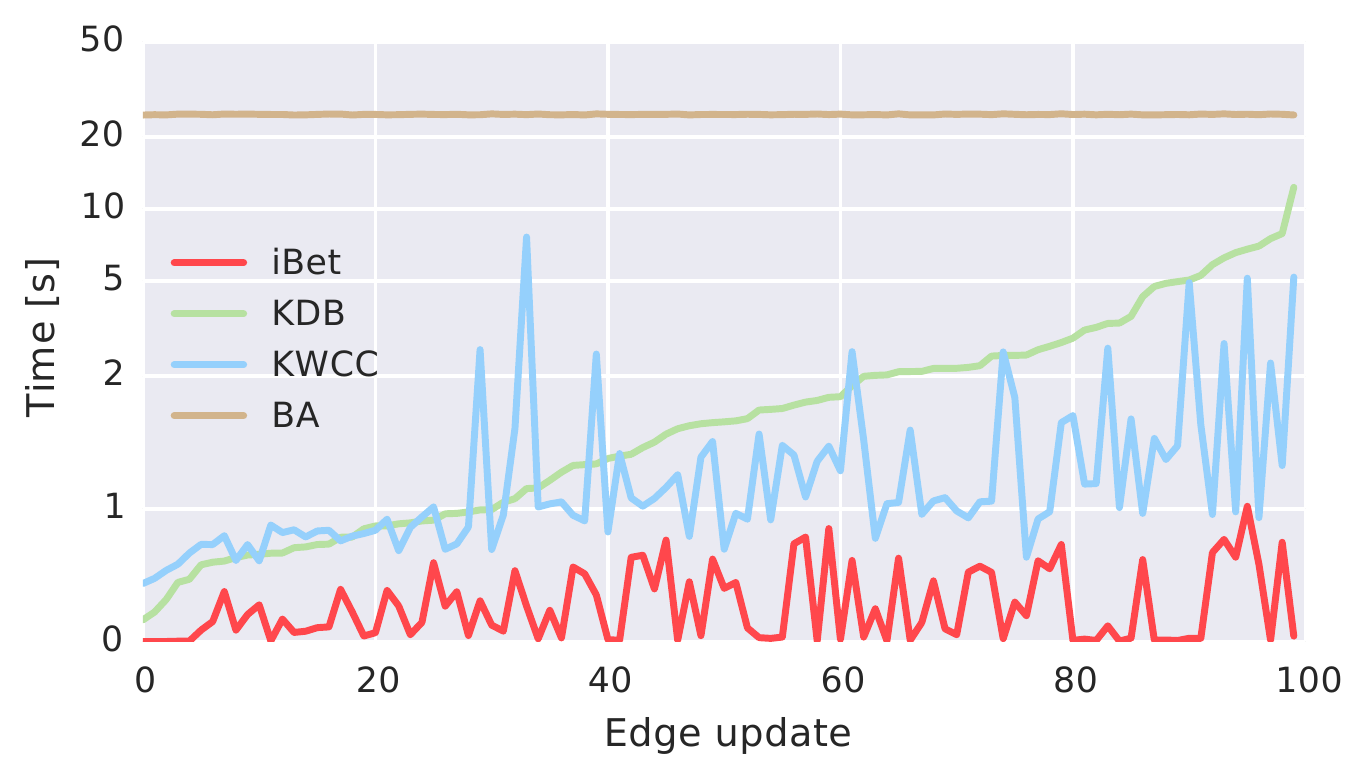}
\includegraphics[width = 0.5\textwidth]{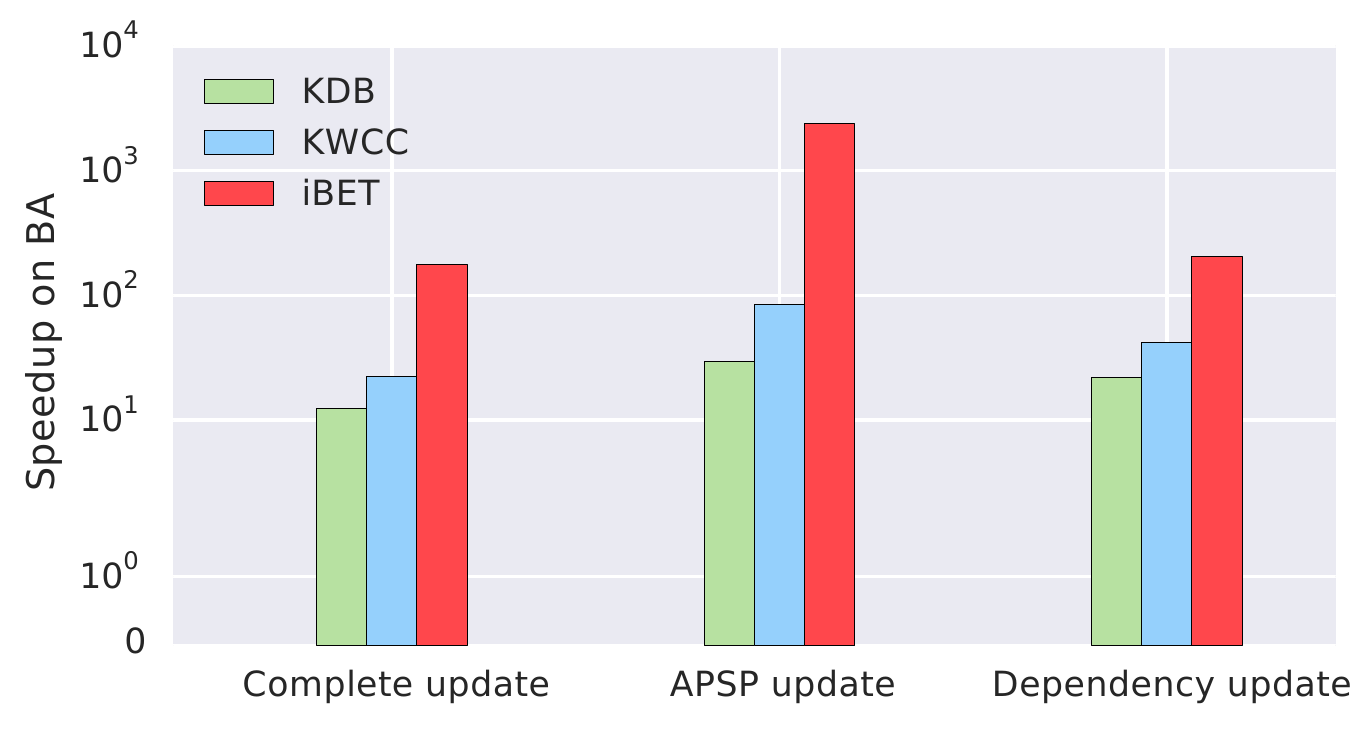}
\caption{Left: Running times of \textsf{iBet}, \textsf{KDB}, \textsf{KWCC} and \textsf{BA} on the \texttt{oregon1-010526} graph for 100 edge updates. Right: Average speedups on recomputation with \textsf{BA} (geometric mean) over all networks of Table~\ref{table:edge_undirected} for the three incremental algorithms. The column on the left shows the speedup of the complete update, the one in the middle the speedup of the APSP update only and the one on the right the speedup of the dependency update only.}
\label{fig:global}
\end{center}
\vspace{-4ex}
\end{figure}

Figure~\ref{fig:global} on the left reports the total running times of \textsf{iBet}, \textsf{KDB}, \textsf{KWCC} and \textsf{BA} on \texttt{oregon1-010526}. Although the running times vary significantly among the updates, \textsf{iBet} is always the fastest among all algorithms. On the contrary, there is not always a clear winner between \textsf{KDB} and \textsf{KWCC}.
On the right, Figure~\ref{fig:global} shows the geometric mean of the speedups on recomputation for the three incremental algorithms, considering the complete update, the APSP update step only and the dependency update step only, respectively. \textsf{iBet} is the method with the highest speedup both overall and on the APSP update and dependency update steps separately, meaning that each of the improvements described in Section~\ref{sec:new-apsp} and Section~\ref{sec:new-dep} contribute to the final speedup. On average, \textsf{iBet} is a factor 82.7 faster than \textsf{KDB} and a factor 28.5 faster than \textsf{KWCC} on the APSP update step and it is a factor 9.4 faster than \textsf{KDB} and a factor 4.9 faster than \textsf{KWCC} on the dependency update step. Overall, the speedup of \textsf{iBet} on \textsf{KDB} ranges from 6.6 to 29.7 and is on average (geometric mean of the speedups) 14.7 times faster. The average speedup on \textsf{KWCC} is 7.4, ranging from a factor 4.1 to a factor 16.0.

\section{Conclusions and future work}
Computing betweenness centrality is a problem of great practical relevance. In this paper we have proposed and evaluated new techniques for the betweenness update after the insertion (or weight decrease) of an edge.
Compared to other approaches, our new algorithm is easy to implement and significantly reduces the number of operations of both the APSP update and the dependency update. Our experiments on real-world networks show that our approach outperforms existing methods, on average approximately by one order of magnitude.

Future work might include parallelization for further acceleration. Furthermore, we plan to extend our techniques also to the decremental case (where an edge can be deleted from the graph or its weight can be increased) and to batch updates, where several edge updates might occur at the same time.

Although dynamic betweenness algorithms can be much faster than recomputation, a major limitation for their scalability is their memory requirement of $\Theta(n^2)$. An interesting research direction is the design of scalable dynamic algorithms with a smaller memory footprint.

Our implementations are based on \textit{NetworKit}~\cite{DBLP:journals/corr/StaudtSM14}, the open-source framework for network analysis, and we will publish our source code in upcoming releases of the package.
\newpage
\bibliographystyle{abbrv}
\bibliography{p46-bergamini}

\end{document}